\documentclass[10pt]{article}
\usepackage[top=1in, bottom=1in, left=1in, right=1in]{geometry}
\usepackage{amsmath, amssymb, amsthm}
\usepackage{color, xcolor, colortbl}
\usepackage{graphics, subfig, overpic}
\usepackage{hyperref}
\usepackage{algorithm, algpseudocode}
\usepackage{multirow}
\usepackage{arydshln}
\usepackage[capitalize,nameinlink]{cleveref} \crefname{equation}{}{Equations}
\Crefname{equation}{Eq.}{Equations}

\graphicspath{{images/}}
\numberwithin{equation}{section}
\newtheorem{theorem}{Theorem}

\newtheorem{proposition}[theorem]{Proposition}

\newcommand\bbR{\mathbb{R}}

\newcommand\bbN{\mathbb{N}}

\newcommand\bbS{\mathbb{S}}

\newcommand\mU{\mathcal{U}}
\newcommand\Nshape{{N_{\mathrm{shape}}}}

\newcommand\dd{\,\mathrm{d}}     \newcommand\unit[1]{\widehat{#1}}

\newcommand\sffont[1]{{\sf{#1}}}
\newcommand\relu{{\sffont{ReLU}}}
\newcommand\id{{\sffont{id}}}
\newcommand\Convone{{\sffont{Conv1d }}}
\newcommand\Convtwo{{\sffont{Conv2d }}}
\newcommand\BCR{{\text{\sffont{BCR-Net}}}}

\newcommand\NN{\mathrm{NN}}
\newcommand\T{\mathsf{T}}
\newcommand\cnn{\mathrm{cnn}}

\title{Solving Inverse Wave Scattering with Deep Learning}
\date{}
\author{
  Yuwei Fan    \thanks{Department of Mathematics, Stanford University, Stanford, CA 94305. 
    Email: {\tt ywfan@stanford.edu}},~~
  Lexing Ying    \thanks{Department of Mathematics and ICME, Stanford University, Stanford, CA 94305.
    Email: {\tt lexing@stanford.edu}}
}

\begin{document}
\maketitle

\begin{abstract}  This paper proposes a neural network approach for solving two classical problems in the
  two-dimensional inverse wave scattering: far field pattern problem and seismic imaging. The
  mathematical problem of inverse wave scattering is to recover the scatterer field of a medium
  based on the boundary measurement of the scattered wave from the medium, which is high-dimensional
  and nonlinear. For the far field pattern problem under the circular experimental setup, a
  perturbative analysis shows that the forward map can be approximated by a vectorized convolution
  operator in the angular direction. Motivated by this and filtered back-projection, we propose an
  effective neural network architecture for the inverse map using the recently introduced BCR-Net
  along with the standard convolution layers. Analogously for the seismic imaging problem, we
  propose a similar neural network architecture under the rectangular domain setup with a
  depth-dependent background velocity.  Numerical results demonstrate the efficiency of the proposed
  neural networks.
\end{abstract}


\section{Introduction}\label{sec:intro}
Inverse wave scattering is the problem of determining the intrinsic property of an object based on the
data collected from the object scatters incoming waves under the illumination of an incident wave,
which can be acoustic, electromagnetic, or elastic. In most cases, inverse wave scattering is
non-intrusive to the object under study and therefore it has a wide range of applications including
radar imaging \cite{borden2001mathematical}, sonar imaging \cite{greene1988acoustical}, seismic
exploration \cite{weglein2003inverse}, geophysics exploration \cite{verschuur1997estimation}, and
medicine imaging \cite{henriksson2010quantitative} and so on.

\paragraph{Background.}
We focus on the time harmonic acoustic inverse scattering in two dimensions. Let $\Omega$ be a
compact domain of interest. The inhomogeneous media scattering problem at a fixed frequency $\omega$
is modeled by the Helmholtz equation
\begin{equation}
  \label{eq:helmholtz}
  Lu := \bigg(-\Delta - \frac{\omega^2}{c^2(x)} \bigg) u,
\end{equation}
where $c(x)$ is the unknown velocity field. Assume that there exists a known background velocity
$c_0(x)$ such that $c(x)$ is identical to $c_0(x)$ outside the domain $\Omega$. By introducing the
{\em scatterer} $\eta(x)$:
\begin{equation}\label{eta def}
  \eta(x) = \frac{\omega^2}{c(x)^2} - \frac{\omega^2}{c_0(x)^2}
\end{equation}
compactly supported in $\Omega$, one can equivalently work with $\eta(x)$ instead of $c(x)$. Note
that in this definition $\eta(x)$ scales quadratically with the frequency $\omega$. However, as
$\omega$ is assumed to be fixed throughout this paper, this scaling does not affect our discussion.

In order to recover the unknown $\eta(\cdot)$, a typical setup of an experiment is as follows. For
each {\em source} $s$ from a source set $S$, one specifies an incoming wave (typically either a
plane wave or a point source) and propagates the wave to the scatterer $\eta(\cdot)$. The scattered
wave field $u^s(x)$ is then recorded at each {\em receiver} $r$ from a receiver set $R$ (typically
placed at the domain boundary or infinity). The whole dataset, indexed by both the source $s$
and the receiver $r$, is denoted by $\{d(s,r)\}_{s\in S, r\in R}$. The forward problem is to compute
$d(s,r)$ given $\eta(x)$. The inverse scattering problem is to recover $\eta(x)$ given $d(s,r)$,

Both the forward and inverse problems are computationally quite challenging, and in the past several
decades a lot of research has been devoted to their numerical solution \cite{erlangga2008advances,
  ernst2012difficult}. For the forward problem, the time-harmonic Helmholtz equation, especially in
the high-frequency regime $\omega\gg 1$, is hard to solve mainly due to two reasons: (1) the
Helmholtz operator has a large number of both positive and negative eigenvalues, with some close to
zero; (2) a large number of degrees of freedom are required for discretization due to the Nyquist
sampling rate. In recent years, quite a few methods have been developed for rapid solutions of
Helmholtz operator \cite{gander2001ailu, engquist2010fast, engquist2011pml, engquist2011hmatrix,
  fei2016recursive}.  The inverse problem is more difficult for numerical solution, due to the
nonlinearity of the problem. For the methods based on optimization, as the loss landscape is highly
non-convex (for example, the cycle skipping problem in seismic imaging \cite{mora1987nonlinear}),
the optimization can get stuck at a local minimum with rather large loss value. Other popular
methods include the factorization method and the linear sampling method
\cite{kirsch1999factorization, cheney2001linear, cakoni2011linear}.

\paragraph{A deep learning approach.}

Deep learning (DL) has recently become the state-of-the-art approach in many areas of machine
learning and artificial intelligence, including computer vision, image processing, and speech
recognition
\cite{Hinton2012,Krizhevsky2012,goodfellow2016deep,MaSheridan2015,Leung2014,SutskeverNIPS2014,leCunn2015,SCHMIDHUBER2015}.
From a technical point of view, this success can be attributed to several key developments: neural
networks (NNs) as a flexible framework for representing high-dimensional functions and maps,
efficient general software packages such as Tensorflow and Pytorch, computing power provided by GPUs
and TPUs, and effective algorithms such as back-propagation (BP) and stochastic gradient descent
(SGD) for tuning the model parameters,

More recently, deep neural networks (DNNs) have been increasingly applied to scientific computing
and computational engineering, particularly for PDE-related problems
\cite{khoo2017solving,berg2017unified,han2018solving,fan2018mnn,Araya-Polo2018, Raissi2018,
  kutyniok2019theoretical, feliu2019meta}. One direction focuses on the low-dimensional
parameterized PDE problems by representing the nonlinear map from the high-dimensional parameters of
the PDE solution using DNNs
\cite{long2018pde,han2017deep,khoo2017solving,fan2018mnn,fan2018mnnh2,fan2019bcr,li2019variational,bar2019unsupervised}. A
second direction takes DNN as an ansatz for high-dimensional PDEs
\cite{rudd2015constrained,carleo2017solving,han2018solving,khoo2019committor, weinan2018deep} since
DNNs offer a powerful tool for approximating high-dimensional functions and densities
\cite{cybenko1989approximation}.

As an important example of the first direction, DNNs have been widely applied to inverse problems
\cite{khoo2018switchnet,hoole1993artificial,kabir2008neural, adler2017solving, lucas2018using,
  tan2018image, fan2019eit, fan2019ot, raissi2019physics}. For the forward problem, as applying
neural networks to input data can be carried out rapidly due to novel software and hardware
architectures, the forward solution can be significantly accelerated once the forward map is
represented with a DNN. For the inverse problem, two critical computational issues are the choices
of the solution algorithm and the regularization term. DNNs can help on both aspects: first,
concerning the solution algorithm, due to its flexibility in representing high-dimensional
functions, DNN can potentially be used to approximate the full inverse map, thus avoiding the
iterative solution process; second, concerning the regularization term, DNNs often can automatically
extract features from the data and offer a data-driven regularization prior.

This paper applies the deep learning approach to inverse wave scattering by representing the whole
inverse map using neural networks. Two cases are considered here: (1) far field pattern and (2)
seismic imaging. In our relatively simple setups, the main difference between the two is the source
and receiver configurations: in the far field pattern, the sources are plane waves and the receivers
are regarded as placed at infinity; in the seismic imaging, both the sources and receivers are
placed at the top surface of the survey domain.

In each case, we start with a perturbative analysis of the forward map, which shows that the forward
map contains a vectorized one-dimensional convolution, after appropriate reparameterization of the
unknown coefficient $\eta(x)$ and the data $d(s,r)$. This observation suggests to represent the
forward map from $\eta$ to $d$ by a \emph{one-dimensional} convolution neural network (with multiple
channels). The filtered back-projection method \cite{norton1979ultrasonic} approximates the inverse
map with the adjoint of the forward map followed by a pseudo-differential filtering. This suggests
an inverse map architecture of reversing the forward map network followed by a simple
two-dimensional convolution neural network. For both cases, the resulting neural networks enjoy a
relatively small number of parameters, thanks to the convolutional structure. This small number of
parameters also allows for accurate and rapid training, even with a somewhat limited dataset.

\paragraph{Organization.}
This rest of the paper is organized as follows. \cref{sec:farfield} discusses the far field pattern
problem. \cref{sec:seismic} considers the seismic imaging problem. \cref{sec:conclusion} concludes
with some discussions for future work.

\section{Far field pattern}\label{sec:farfield}

\subsection{Mathematics analysis}

In the far field pattern case, the background velocity field $c_0(x)$ is constant, and without loss
of generality, equal to one. We introduce the base operator
$L_0=-\Delta-\omega^2/c_0^2=-\Delta-\omega^2$ and write $L$ in a perturbative way as
\begin{equation}
  L = L_0 - \eta.
\end{equation}
The sources are parameterized by $s\in S=[0,2\pi)$. For each source $s$, the incoming wave is a
plane wave $e^{i\omega \unit{s}\cdot x}$ with the unit direction given by
$\unit{s}=(\cos(s),\sin(s))\in\bbS^1$. The scattered wave $u^s(x)$ satisfies the following equation
\begin{equation}
  \label{eqn:usf}
  (L_0 - \eta) (e^{i\omega \unit{s}\cdot x} + u^s(x)) = 0,
\end{equation}
along with the Sommerfeld radiation boundary condition at infinity \cite{colton1998inverse}.
The receivers are also indexed by $r\in R=[0,2\pi)$.  The far field pattern at the unit direction
$\unit{r}=(\cos(r), \sin(r)) \in\bbS^1$ is defined as
\[
  \widehat{u}^s(r) \equiv \lim_{\rho\rightarrow\infty} \sqrt{\rho} e^{-i\omega\rho}
  u^s(\rho\cdot\unit{r}).
\]
The recorded data is the set of far field pattern from all incoming directions:
$d(s,r)\equiv\widehat{u}^s(r)$ for $r\in R$ and $s\in S$.

In order to understand better the relationship between $\eta(x)$ and $d(s,r)$, we perform a
perturbative analysis for small $\eta$. Expanding \eqref{eqn:usf} leads to 
\[
(L_0 u^s)(x) = \eta(x) e^{i\omega \unit{s}\cdot x} + \ldots,
\]
where $\ldots$ stands for higher order terms in $\eta$. Letting $G_0 = L_0^{-1}$ be the Green's
functions of the free-space Helmholtz operator $L_0$, we get
\[
  u^s(y) = \int G_0(y-x) \eta(x) e^{i\omega \unit{s}\cdot x} \dd x  + \ldots
\]
Using the expansion at infinity
\[
G_0(z) = \frac{ e^{i\omega |z|} + o(1) }{\sqrt{|z|}},
\]
we arrive at
\begin{equation} \label{eq:d1}
  \widehat{u}^s(r) = \lim_{\rho\rightarrow\infty} \sqrt{\rho} e^{-i\omega \rho} u^s(\rho \cdot \unit{r})  
  \approx \lim_{\rho\rightarrow\infty} \int \frac{ e^{i\omega (\rho-\unit{r}\cdot x)}  }{\sqrt{\rho}}  
  \sqrt{\rho} e^{-i\omega \rho}\eta(x)
  e^{i\omega \unit{s}\cdot x}  \dd x \\
  = \int e^{-i\omega(\unit{r}-\unit{s})\cdot x} \eta(x) \dd x \equiv d_1(s,r),
              \end{equation}
where the notation $d_1(s,r)$ stands for the first order approximation to $d(s,r)$ in $\eta$.

\subsubsection{Problem setup}
\begin{figure}[!ht]
  \centering
  \includegraphics[width=0.3\textwidth]{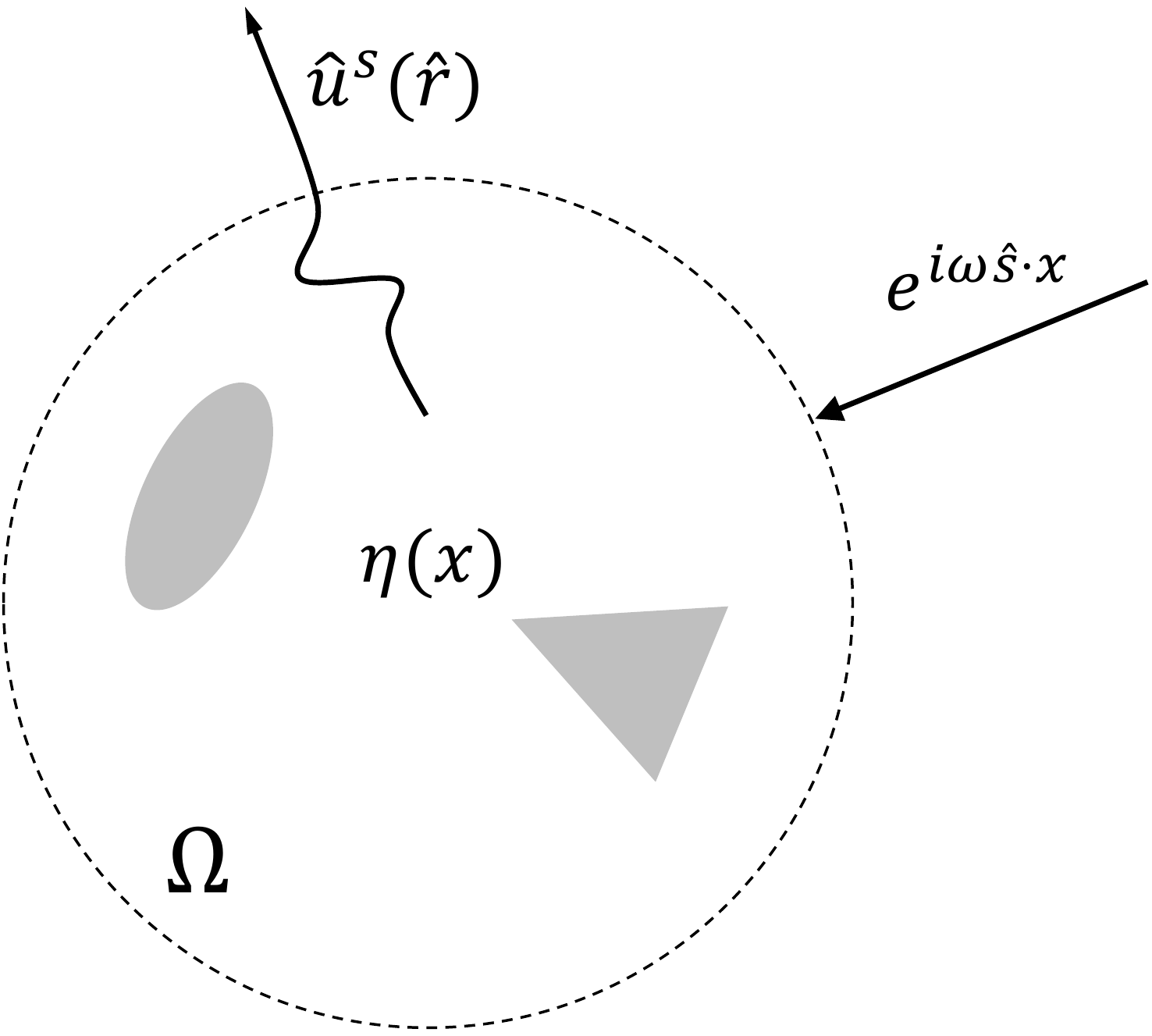}
  \caption{\label{fig:domainSca} Illustration of the incoming and outgoing waves for a far field
  pattern problem. The scatterer $\eta(x)$ is compactly supported in the domain $\Omega$. The
  incoming plane wave points at direction $\unit{s}=(\cos(s),\sin(s))$. The far field pattern is sampled at each
  receiver direction $\unit{r}=(\cos(r), \sin(r))$.
  }
\end{figure}

For the far field pattern problem, we are free to treat the domain $\Omega$ as the unit disk
centered at origin (by appropriate rescaling and translation), as illustrated in
\cref{fig:domainSca}.  In a commonly used setting, the sources and receivers are uniformly sampled
in $\bbS^1$: $s=\frac{2\pi j}{N_s}$, $j=0, \dots,N_s-1$ and $r=\frac{2\pi k}{N_r}$, $k=0, \dots,
N_r-1$, where $N_s=N_r$ for simplicity.

\begin{figure}[ht]
  \centering
  \begin{tabular}{c@{}c@{}c}
    \subfloat[$\eta(x)$]{
      \includegraphics[width=0.2\textwidth]{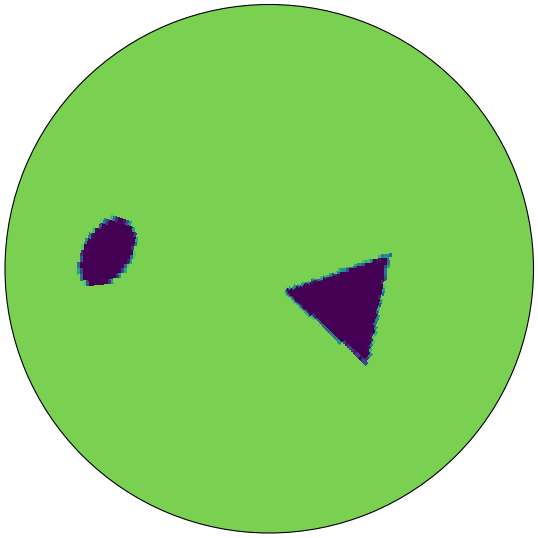} 
    } &
    \subfloat[$\Re(d(s, r))$]{
      \includegraphics[width=0.3\textwidth]{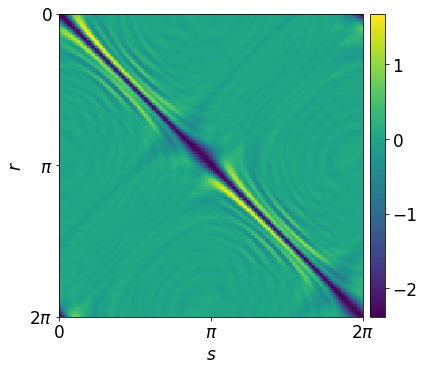} 
    } &
    \subfloat[$\Im(d(s, r))$]{
      \includegraphics[width=0.3\textwidth]{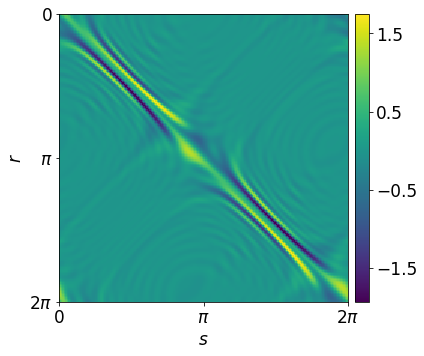} 
    }\\
    \subfloat[$\eta(\theta, \rho)$]{
      \includegraphics[width=0.28\textwidth]{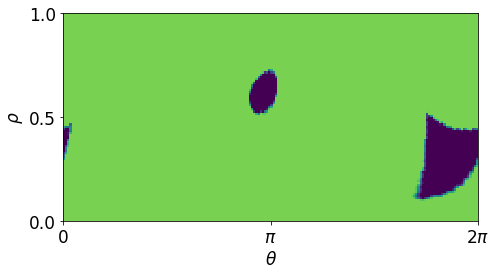} 
    } &
    \subfloat[$\Re(d(m, h))$]{
      \includegraphics[width=0.3\textwidth]{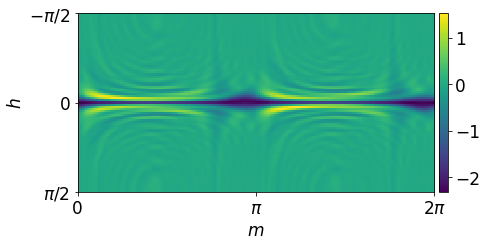} 
    } &
    \subfloat[$\Im(d(m, h))$]{
      \includegraphics[width=0.3\textwidth]{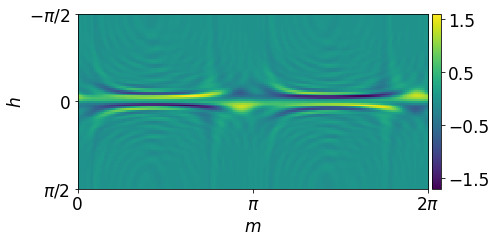} 
    }
  \end{tabular}
  \caption{\label{fig:measurement} Visualization of the scatterer field $\eta$ and the measurement
    data $d$.  The upper figures are the scatterer and the real and imaginary parts of the
    measurement data $d$, respectively. The lower-left figure is the scatterer in the polar
    coordinates. The next two lower figures are the real and imaginary parts of the measurement data
    after the change of variable, respectively.  }
\end{figure}

\subsubsection{Forward map}
Since the domain $\Omega$ is the unit disk, it is convenient to work with the problem in the polar
coordinates. Let $x=(\rho\cos(\theta), \rho\sin(\theta))$, where $\rho\in[0, 1]$ is the radial
coordinate and $\theta\in[0,2\pi)$ is the angular one. Due to the circular tomography geometry that
  $r,s\in[0, 2\pi)$, it is convenient to reparameterize the measurement data by a change of
    variables
\begin{equation}
  m = \frac{r+s}{2}, h = \frac{r-s}{2},\quad\text{or equivalently}\quad r = m+h, s = m-h,
\end{equation}
where all the variables $m, h, r, s$ are understood modulus $2\pi$.  \Cref{fig:measurement} presents
an example of the scatterer field $\eta(x)$ and the measurement data $d(s,r)$ in the original and
transformed coordinates.

With a bit abuse of notation, we can redefine the measurement data
\begin{equation}
  d(m,h) \equiv d(s, r)|_{s=m-h, r=m+h},
\end{equation}
and so does $d_1(m, h)$. At the same time, we redefine
\[
\eta(\theta,\rho) = \eta(\rho\cos(\theta),\rho\sin(\theta)
\]
in the polar coordinates.  Since the first order approximation $d_1(m, h)$ is linearly dependent on
$\eta(\theta,\rho)$, there exists a kernel distribution $K(m, h, \theta,\rho)$ such that
\begin{equation}
  d_1(m, h) = \int_{0}^1\int_{0}^{2\pi}K(m, h, \theta, \rho)\eta(\theta, \rho)\dd\rho\dd\theta.
\end{equation}
Since the domain is the unit disk centered at origin and the background velocity field $c_0=1$ is
constant, the whole problem is equivariant to rotation. In this case, the system can be dramatically
simplified due to the following proposition.
\begin{proposition}\label{pro:convolution}
  There exists a function $\kappa(h, \rho, \cdot)$ periodic in the last parameter such that $K(m, h,
  \theta, \rho) = \kappa(h, \rho, m-\theta)$ or equivalently,
  \begin{equation}\label{eq:convolution}
    d_1(m,h) = \int_0^1 \int_0^{2\pi} \kappa(h,\rho,m-\theta) \eta(\rho,\theta) \dd\theta \dd\rho.
  \end{equation}
\end{proposition}
\begin{proof}
  A simple calculation shows that the phase $(\unit{r}-\unit{s})\cdot x$ now becomes
  \[
    (\unit{r}-\unit{s})\cdot x = 
    (\unit{(m+h)}-\unit{(m-h)})\cdot (\rho\cos(\theta),\rho\sin(\theta)) = 
    2\rho \sin(h) \sin(\theta-m).
  \]
  Therefore, \cref{eq:d1} turns to
  \[
  d_1(m,h) = \int_0^1  
  \int_{0}^{2\pi} \left(\rho e^{2i\rho\omega \sin h \sin(m-\theta)}\right) \eta(\theta, \rho)\dd\theta \dd\rho.
  \]
  Setting $\kappa(h,\rho,y) = \rho e^{2i\rho\omega \sin h \sin(y)}$ completes the proof.
\end{proof}

\Cref{pro:convolution} shows that $K$ acts on $\eta$ in the angular direction by a convolution,
which allows us to evaluate the map $\eta(\theta,\rho)\to d_1(m, h)$ by a family of 1D convolutions,
parameterized $\rho$ and $h$.

\paragraph{Discretization.}
Until now, the discussion is in the continuous space. The discretization and solution method of the
Helmholtz equation will be discussion along with the numerical results. With a slight abuse of
notation, the same symbols will be used to denote the discrete version of the continuous data and
kernels and the discrete version of \cref{eq:d1,eq:convolution} takes the form
\begin{equation}\label{eq:discrete}
  d(m,h) \approx \sum_{\rho,\theta}K(m, h, \theta, \rho)\eta(\theta, \rho) 
  = \sum_\rho \left( \kappa(h,\rho,\cdot) * \eta(\cdot, \rho) \right)(m).
\end{equation}

\subsection{Neural network}\label{sec:NN}

\paragraph{Forward map.}
The perturbative analysis shows that, when $\eta$ is sufficiently small, the forward map
$\eta(\theta,\rho)\to d(m, h)$ can be approximated with \cref{eq:discrete}. In terms of the NN
architecture, for small $\eta$, by taking $h$ and $\rho$ directions as the channel dimension, the
forward map \cref{eq:discrete} can be approximated by a one-dimensional (non-local) convolution
layer. For larger $\eta$, this linear approximation is no longer accurate, which can be addressed 
by increasing the number of convolution layers and including nonlinear activations for the neural
network of \cref{eq:discrete}.

\begin{algorithm}[ht]
  \begin{small}
    \begin{center}
      \begin{algorithmic}[1]
        \Require {$c$, $N_{\cnn}\in\bbN^+$, $\eta\in\bbR^{N_{\theta}\times N_{\rho}}$}
        \Ensure {$d \in\bbR^{N_s\times N_h}$}
        \State $\xi = \Convone[c,1,\id](\eta)$ with $\rho$ as the channel direction
        \Comment{Resampling $\eta$ to fit for \BCR}
        \State $\zeta = \BCR[c, N_{\cnn}](\xi)$
        \Comment{Use {\BCR} to implement the convolutional neural network.}
        \State $d = \Convone[N_{h}, 1, \id](\zeta)$
        \Comment{Reconstruct the result $d$ from the output of \BCR.}
        \State \sffont{return} $d$
      \end{algorithmic}
    \end{center}
  \end{small}
  \caption{\label{alg:forward} Neural network architecture for the forward map $\eta\to d$.}
\end{algorithm}

The number of channels, denoted by $c$, is quite problem-dependent and will be discussed in the
numerical section. Notice that the convolution between $\eta$ and $d$ is global in the angular
direction. In order to represent global interactions, the window size $w$ of the convolution layer
and the number of layers $N_{\cnn}$
must satisfy the constraint
\begin{equation}
  w N_{\cnn}\geq N_{\theta},
\end{equation}
where $N_{\theta}$ is number of discretization points on the
angular direction. A simple calculation shows that the number of parameters of the neural network is
$O(wN_{\cnn}c^2)\sim O(N_{\theta}c^2)$. The recently introduced {\BCR} \cite{fan2019bcr} has been
demonstrated to require fewer number of parameters and provide good efficiency for such
interactions. Therefore, we replace the convolution layers with the {\BCR} in our architecture. The
resulting neural network architecture for the forward map is summarized in \cref{alg:forward} with
an estimate of $O(c^2\log(N_{\theta})N_{\cnn})$ parameters. The components of 
\cref{alg:forward} are detailed below.
\begin{itemize}
\item $\xi=\Convone[c, w, \phi](\eta)$, which maps $\eta\in \bbR^{N_{\theta}\times N_{\rho}}$ to
  $\xi\in\bbR^{N_{\theta}\times c}$ is the one-dimensional convolution layer with window size $w$,
  channel number $c$, activation function $\phi$ and period padding on the first direction.
\item {\BCR} is a multiscale neural network motivated by the data-sparse nonstandard wavelet representation of the
  linear operators \cite{bcr}. It processes the information at different scale separately and each
  scale can be understood as a {\em local} convolutional neural network. The one-dimensional
  $\zeta=\BCR[c, N_{\cnn}](\xi)$ maps $\xi\in\bbR^{N_{\theta}\times c}$ to
  $\zeta\in\bbR^{N_{\theta}\times c}$ where $c$ is the number of wavelets in each scale and
  $N_{\cnn}$ denotes the number of layers in the local convolutional neural network in each scale.
  The readers are referred to \cite{fan2019bcr} for more details on the \BCR.
\end{itemize}

\paragraph{Inverse map.}
As we have seen, if $\eta$ is sufficiently small, the forward map can be approximated by $d \approx
K \eta$, the operator notation of the discretization \cref{eq:discrete}. Here $\eta$ is a vector
indexed by $(\theta, \rho)$, $d$ is a vector indexed by $(m, h)$, and $K$ is a matrix with row
indexed by $(m, h)$ and column indexed by $(\theta, \rho)$.

The filtered back-projection method \cite{norton1979ultrasonic} suggests the following formula to
recover $\eta$:
\begin{equation}
  \eta\approx (K^\T K + \epsilon I)^{-1} K^\T d,
\end{equation}
where $\epsilon$ is a regularization parameter. The first piece $K^\T d$ can also be written as a
family of convolutions as well
\begin{equation}
  (K^\T d)(\theta,\rho) = \sum_{h} (\kappa(h, \rho, \cdot) * d(\cdot, h))(\theta).
\end{equation}
The application of $K^\T$ to $d$ can be approximated with a neural network similar to the one for
$K$ in \cref{alg:forward}, by reversing the order. The second piece $(K^\T K + \epsilon I)^{-1}$ is
a pseudo-differential operator in the $(\theta, \rho)$ space and it is implemented with several
two-dimensional convolutional layers for simplicity. Putting two pieces together, the resulting
architecture for the inverse map is summarized in \cref{alg:inverse} and illustrated in
\cref{fig:inverse}. Here, $\Convtwo[c_2, w, \phi]$ used in \cref{alg:inverse} is a two-dimensional
convolution layer with window size $w$, channel number $c_2$, activation function $\phi$ and
periodic padding on the first direction and zero padding on the second direction. The selection of
the hyper-parameters in \cref{alg:inverse} will be discussed in the numerical section.

\begin{algorithm}[htb]
  \begin{small}
    \begin{center}
      \begin{algorithmic}[1]
        \Require {$c, c_2$, $w$, $N_{\cnn}$, $N_{\cnn2}\in\bbN^+$, $d\in\bbR^{N_{s}\times N_h}$}
        \Ensure {$\eta\in\bbR^{N_{\theta}\times N_{\rho}}$}
        \Statex \# \textit{Application of $K^\T$ to $d$}
        \State $\zeta = \Convone[c, 1, \id](d)$ with $h$ as the channel direction
        \State $\xi = \BCR[c, N_{\cnn}](\zeta)$
        \State $\xi^{(0)} = \Convone[N_\rho, 1, \id](\xi)$
        \Statex \# \textit{Application of $(K^\T K+\epsilon I)^{-1}$}
        \For{$k$ from $1$ to $N_{\cnn2}-1$}
        \State $\xi^{(k)} = \Convtwo[c_2, w, \relu](\xi^{(k-1)})$
        \EndFor
        \State $\eta = \Convtwo[1, w, \id](\xi^{(N_{\cnn2}-1)})$
        \State \sffont{return} $d$
      \end{algorithmic}
    \end{center}
  \end{small}
  \caption{\label{alg:inverse} Neural network architecture for the inverse problem $d \to \eta$.}
\end{algorithm}
\begin{figure}[ht]
  \centering
  \includegraphics[width=0.9\textwidth]{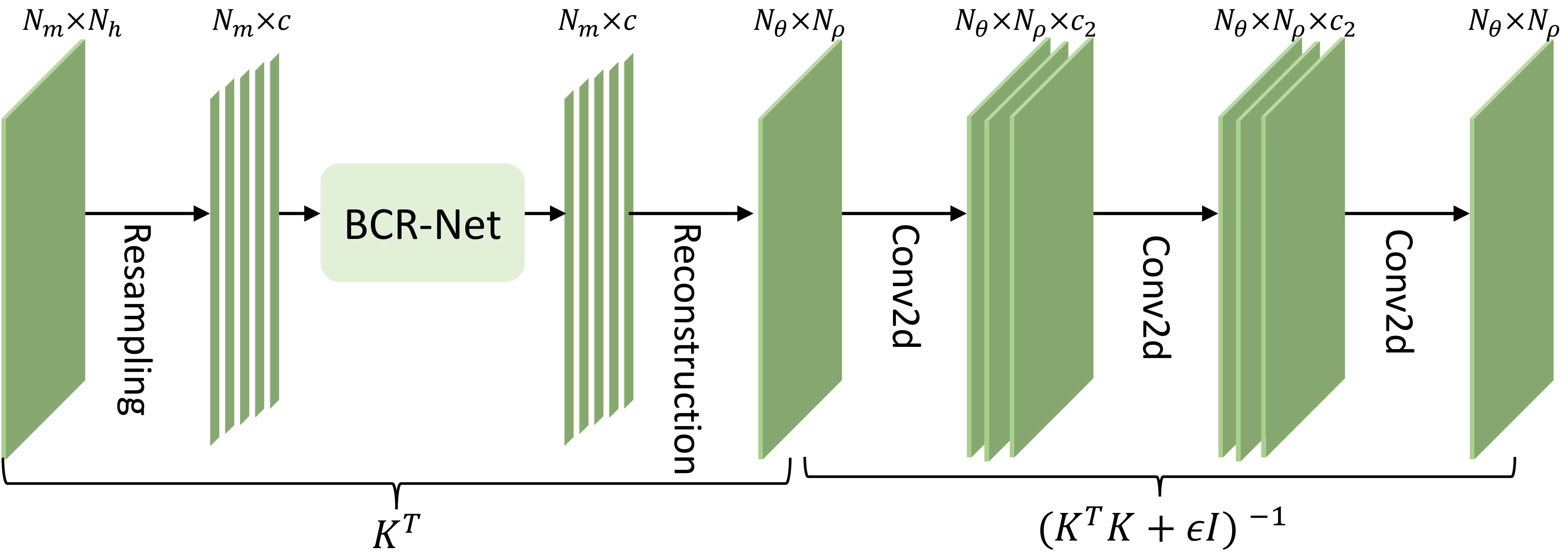}
  \caption{ \label{fig:inverse} Neural network architecture for the inverse map of far field pattern
  problem.}
\end{figure}

\subsection{Numerical examples}
This section reports the numerical setup and results of the proposed neural network architecture in
\cref{alg:inverse} for the inverse map $d\to \eta$.

\subsubsection{Experimental setup}
Since the scatterer $\eta$ is compactly supported in the unit disk $\Omega$, we embed $\Omega$ into
the square domain $[-1, 1]^2$ and solve the Helmholtz equation \eqref{eq:helmholtz} in the
square. In the numerical solution of the Helmholtz equation, we discretize $[-1,1]^2$ with a uniform
Cartesian mesh with $192$ points in each direction by a finite difference scheme (frequency
$\omega=16$ and $12$ points used per wavelength). The perfectly matched layer
\cite{berenger1996perfectly} is used to deal with the Sommerfeld boundary condition and the solution
of the discrete system can be accelerated with appropriate preconditioners (for example,
\cite{engquist2011pml}).

In the polar coordinates of $\Omega$, $(\theta, \rho) \in [0,2\pi)\times [0,1]$ is partitioned by
uniformly Cartesian mesh with $192\times 96$ points, i.e., $N_\theta=192$ and $N_{\rho}=96$.
Given the values of $\eta$ in the Cartesian grid, the values $\eta(\theta,\rho)$ used in
\cref{alg:inverse} in the polar coordinates are computed via linear interpolation.

The number of sources and receivers are $N_s=N_r=192$. The measurement data $d(s,r)$ is generated by
solving the Helmholtz equation $N_s$ times with different incident plane wave.  For the change of
variable of $(s, r)\to (m, h)$, linear interpolation is used to generate the data $d(m, h)$ from
$d(s, r)$. In the $(m,h)$ space, $N_m=192$ for $m\in[0,2\pi)$ and $N_h=96$ for $h\in(-\pi/2,
  \pi/2)$.  Since the measurement data is complex, the real and imaginary parts can be treated
  separately as two channels. The actual simulation suggests that using only the real part (or the
  imaginary part) as input for \cref{alg:inverse} is enough to generate good results.  

The NN in \cref{alg:inverse} is implemented in Keras \cite{keras} on top of TensorFlow \cite{tensorflow}.
The loss function is taken to be the mean squared error and the optimizer used is Nadam
\cite{dozat2015incorporating}. The parameters of the network are initialized by Xavier
initialization \cite{glorot2010understanding}.  Initially, the batch size and the learning rate is
firstly set as $32$ and $10^{-3}$, respectively, and the NN is trained with $100$ epochs. We then
increase the batch size by a factor of $2$ till $512$ with the learning rate unchanged, and next
decrease the learning rate by a factor $10^{1/2}$ down to $10^{-5}$ with the batch size fixed at
$512$. For each batch size and learning rate configuration, the NN is trained with $50$ epochs.  The
hyper-parameters used for \cref{alg:inverse} are $N_{\cnn}=6$, $N_{\cnn2}=5$, and $w=3\times 3$. The
selection of the channel number $c$ will be studied next.

\subsubsection{Results}
For a fixed $\eta$, $d(m, h)$ stands for the \emph{exact} measurement data solved by numerical
discretization of \cref{eq:helmholtz}. The prediction of the NN from $d(m,h)$ is denoted by
$\eta^{\NN}$.  The metric for the prediction is the peak signal-to-noise ratio (PSNR), which is
defined as
\begin{equation}\label{eq:psnr}
  \mathrm{PSNR} = 
  10 \log_{10}\left(\frac{\mathrm{Max}^2}{\mathrm{MSE}}\right),
  ~~\mathrm{Max} = \max_{ij}(\eta_{ij}) - \min_{ij}(\eta_{ij}),
  ~~\mathrm{MSE} = \frac{1}{N_{\theta}N_{\rho}}\sum_{i,j}|\eta_{i,j}-\eta_{i,j}^{\NN}|^2.
\end{equation}
For each experiment, the test PSNR is then obtained by averaging \cref{eq:psnr} over a given set of
test samples. The numerical results presented below are obtained by repeating the training process
five times, using different random seeds for the NN initialization.

The numerical experiments focus on the shape reconstruction setting \cite{kirsch1999factorization,
  kirsch2008factorization, colton2018looking}, where $\eta$ are often piecewise constant inclusions.
Here, the scatterer field $\eta$ is assumed to be the sum of $\Nshape$ piecewise constant shapes.
For each shape, it can be either triangle, square or ellipse, its direction is uniformly random over
the unit circle, its position is uniformly sampled in the disk, and its inradius is sampled from the
uniform distribution $\mU(0.1, 0.2)$. When a shape is an ellipse, the width and height are sampled from
the uniform distribution $\mU(0.1, 0.2)$ and $\mU(0.05, 0.1)$. It is also required that each shape
lies in the disk and there is no intersection between every two shapes. We generate two dataset for
$\Nshape=2$ and $\Nshape=4$, and each has $20,480$ samples $\{(\eta_i, d_i)\}$ with $16,384$ used
for training and the remaining $4,096$ for testing.

\begin{figure}[h!]
  \centering
  \includegraphics[width=0.4\textwidth]{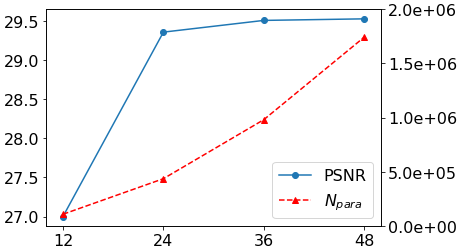}
  \caption{\label{fig:trErr} The test PSNR for different channel numbers $c$ for the dataset
    $\Nshape=4$.  
  }
\end{figure}

\begin{figure}[h!]
  \centering
  \begin{tabular}{l@{}cccc}
    & reference & $\eta^{\NN}$ with $\delta=0$ &
    $\eta^{\NN}$ with $\delta=10\%$ &
    $\eta^{\NN}$ with $\delta=100\%$\\
        \rotatebox{90}{\phantom{ab}\parbox{3cm}{\phantom{$\Nshape$\\}$\Nshape=4$}}\phantom{a} &
    \includegraphics[width=0.22\textwidth]{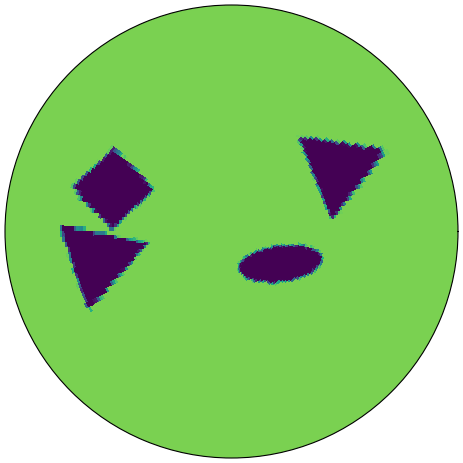}       &
    \includegraphics[width=0.22\textwidth]{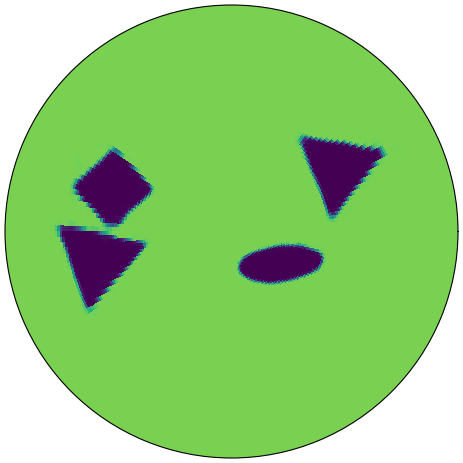}   &
    \includegraphics[width=0.22\textwidth]{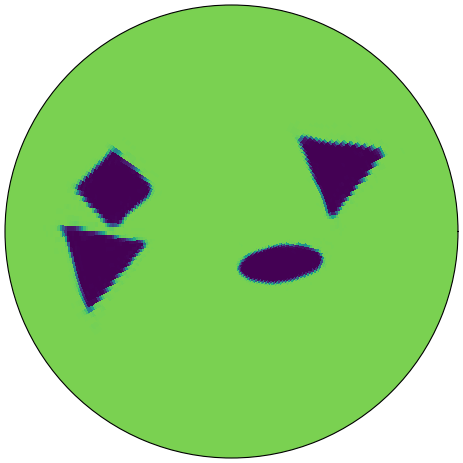} & 
    \includegraphics[width=0.22\textwidth]{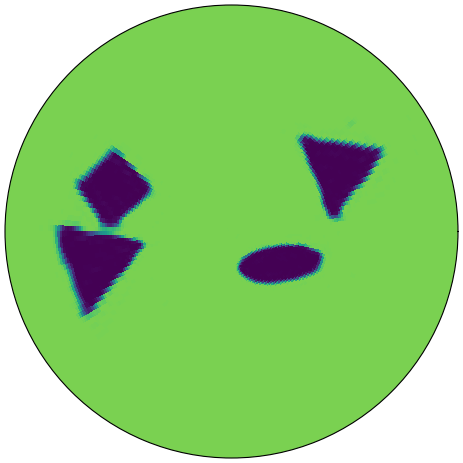} \\ 

        \rotatebox{90}{\phantom{ab}\parbox{3cm}{\phantom{$\Nshape$\\}$\Nshape=2$}}\phantom{a} &
    \includegraphics[width=0.22\textwidth]{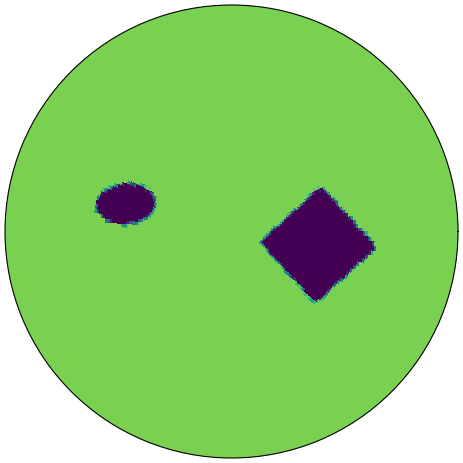}       &
    \includegraphics[width=0.22\textwidth]{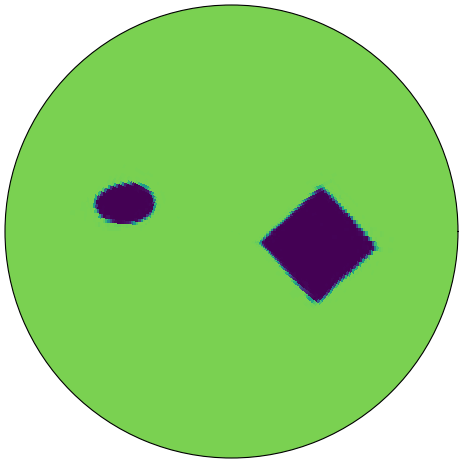}   &
    \includegraphics[width=0.22\textwidth]{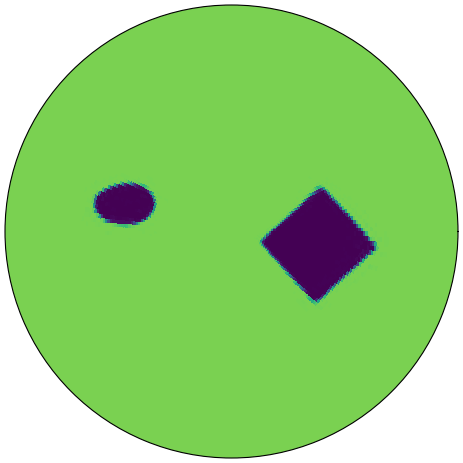} &
    \includegraphics[width=0.22\textwidth]{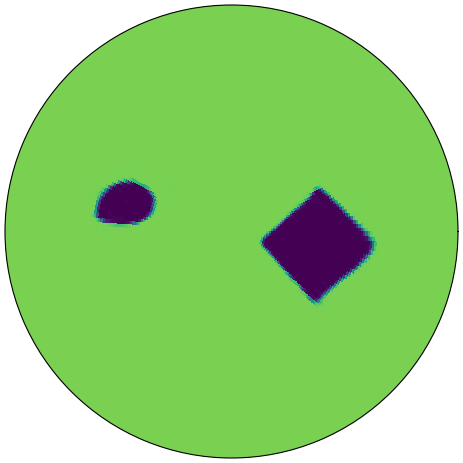} 
  \end{tabular}
  \caption{\label{fig:prediction} The far field pattern problem. NN prediction of a sample with test
    data $\Nshape=4$ (first row) and a sample with $\Nshape=2$ (second row), at different noise
    level $\delta=0$, $10\%$ and $100\%$.  }
\end{figure}

\begin{figure}[h!]
  \centering
  \begin{tabular}{l@{}cccc}
    & reference & $\eta^{\NN}$ with $\delta=0$ &
    $\eta^{\NN}$ with $\delta=10\%$ &
    $\eta^{\NN}$ with $\delta=100\%$\\
    \rotatebox{90}{\phantom{ab}\parbox{3cm}{Train: $\Nshape=2$\\ Test: \phantom{n}$\Nshape=4$}}\phantom{a} &
    \includegraphics[width=0.22\textwidth]{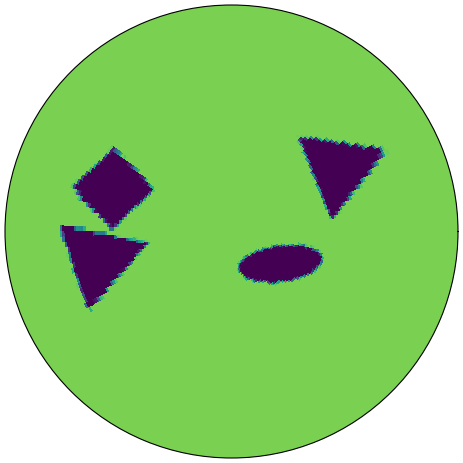}       &
    \includegraphics[width=0.22\textwidth]{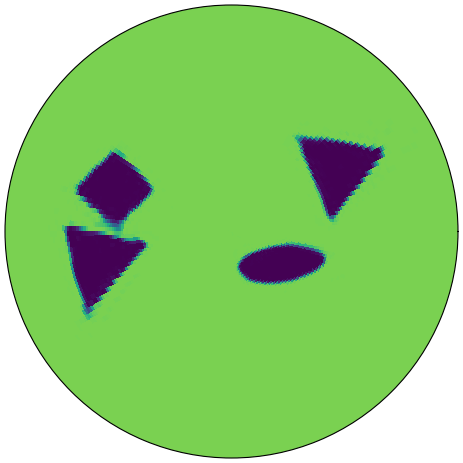}   &
    \includegraphics[width=0.22\textwidth]{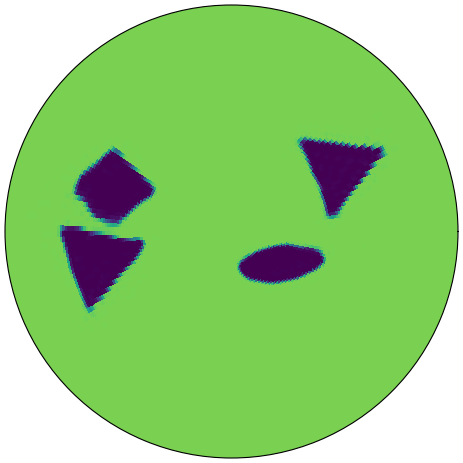} & 
    \includegraphics[width=0.22\textwidth]{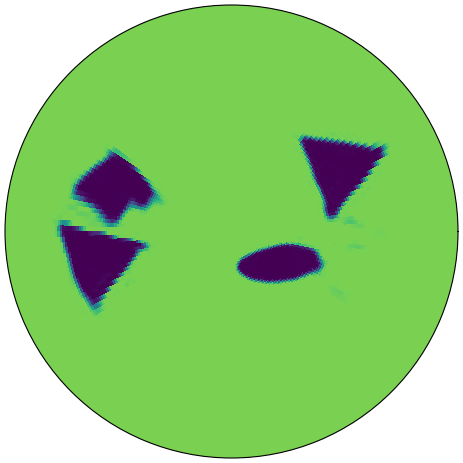} \\ 

    \rotatebox{90}{\phantom{ab}\parbox{3cm}{Train: $\Nshape=4$\\ Test: \phantom{n}$\Nshape=2$}}\phantom{a} &
    \includegraphics[width=0.22\textwidth]{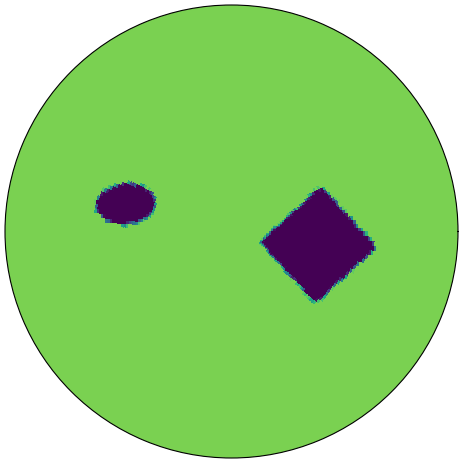}       &
    \includegraphics[width=0.22\textwidth]{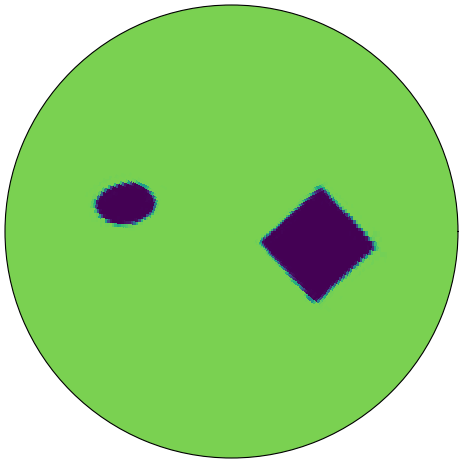}   &
    \includegraphics[width=0.22\textwidth]{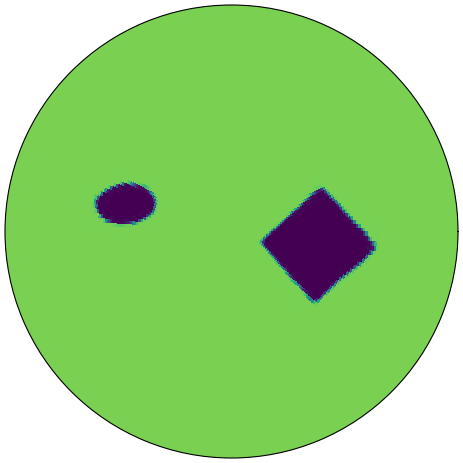} &
    \includegraphics[width=0.22\textwidth]{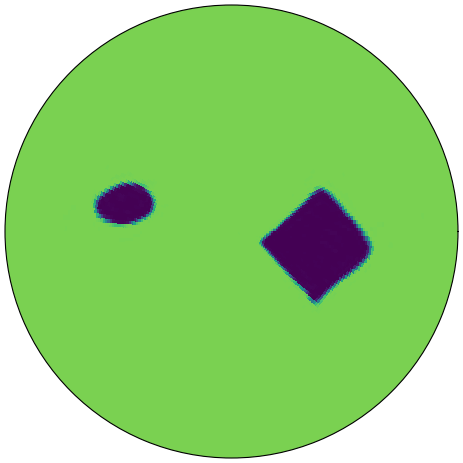} 
  \end{tabular}
  \caption{\label{fig:generalization} The far field pattern problem: NN generalization test. In the
    first row, the NN is trained by the data with the number of shapes $\Nshape=2$ with noise level
    $\delta=0$, $10\%$ or $100\%$ and tested by the data with $\Nshape=4$ at the same noise
    level. In the second row, the NN is trained by the data with $\Nshape=4$ and tested by the data
    with $\Nshape=2$.  }
\end{figure}

We first study the choice of channel number $c$ in \cref{alg:inverse}. \Cref{fig:trErr} presents
the test PSNR and the number of parameters for different channel number $c$ for the dataset
$\Nshape=4$.  As the channel number $c$ increases, the test PSNR first increases consistently and
then saturates. Note that the number of parameters of the neural network is
$O(c^2\log(N_{\theta})N_{\cnn})$. The choice of $c=24$ offers a reasonable balance between accuracy
and efficiency, and the total number of parameters is $439$K.

To model the uncertainty in the measurement data, we introduce noises to the measurement data by
defining $d_i^{\delta}\equiv (1+ Z_i \delta)d_i$, where $Z_i$ is a Gaussian random variable with
zero mean and unit variance and $\delta$ controls the signal-to-noise ratio. 
For each noisy level $\delta=0$, $10\%$, $100\%$, an independent NN is trained and tested with the
noisy dataset $\{(d_i^\delta,\eta_i)\}$.

\Cref{fig:prediction} collects, for different noise level $\delta=0$, $10\%$, $100\%$, samples for
different $\Nshape=2$, $4$.  The NN is trained with the datasets generated in the same way as the
test data. When there is no noise in the measurement data, the NN consistently gives accurate
predictions of the scatterer field $\eta$, in the position, shape, and direction of the shapes. In
particular, for the case $\Nshape=4$, the square in the left part of the domain is close to a
triangle. The NN is able to distinguish the shapes and gives a clear boundary of each.  For the
small noise levels, for example, $\delta=10\%$, the boundary of the shapes slightly blurs while the
position, direction and shape are still correct. As the noise level $\delta$ increases, the boundary
of the shapes blurs more, but the position and direction of shape are always correct. 

The next test is about the generalization of the proposed NN. We first train the NN with one data
set ($\Nshape=2$ or $4$) with noise level $\delta=0$, $10\%$ or $100\%$ and test with the other
($\Nshape=4$ or $2$) with the same noise level. The results, presented in
\cref{fig:generalization}, indicate that the NN trained by the data with two inclusions is capable
of recovering the measurement data of the case with four inclusions, and vice versa. Moreover, the
prediction results are comparable with those in \cref{fig:prediction}. This shows that the trained
NN is capable of predicting beyond the training scenario.

\section{Seismic imaging}\label{sec:seismic}

\subsection{Mathematics analysis}
In the seismic imaging case, $\Omega$ is a rectangular domain with Sommerfeld radiation boundary
condition specified, as illustrated in \cref{fig:seismic}. Following \cite{hulme2004general,
  palermo2016engineered}, we apply periodic boundary conditions in the horizontal direction to our
problem for simplicity. This setup is also appropriate for studying periodic material, such as
phononic crystals \cite{palermo2016engineered, hussein2014dynamics}, etc. After appropriate
rescaling, we consider the domain $\Omega=[0, 1]\times [0, Z]$, where $Z$ is a fixed constant. Both
the sources $S=\{x_s\}$ and the receivers $R=\{x_r\}$ are a set of uniformly sampled points along a
horizontal line near the top surface of the domain, and $x_r=(r, Z)$ and $x_s=(s, Z)$, for $r,s\in
[0,1]$.

\begin{figure}[!ht]
  \centering
  \includegraphics[width=0.45\textwidth]{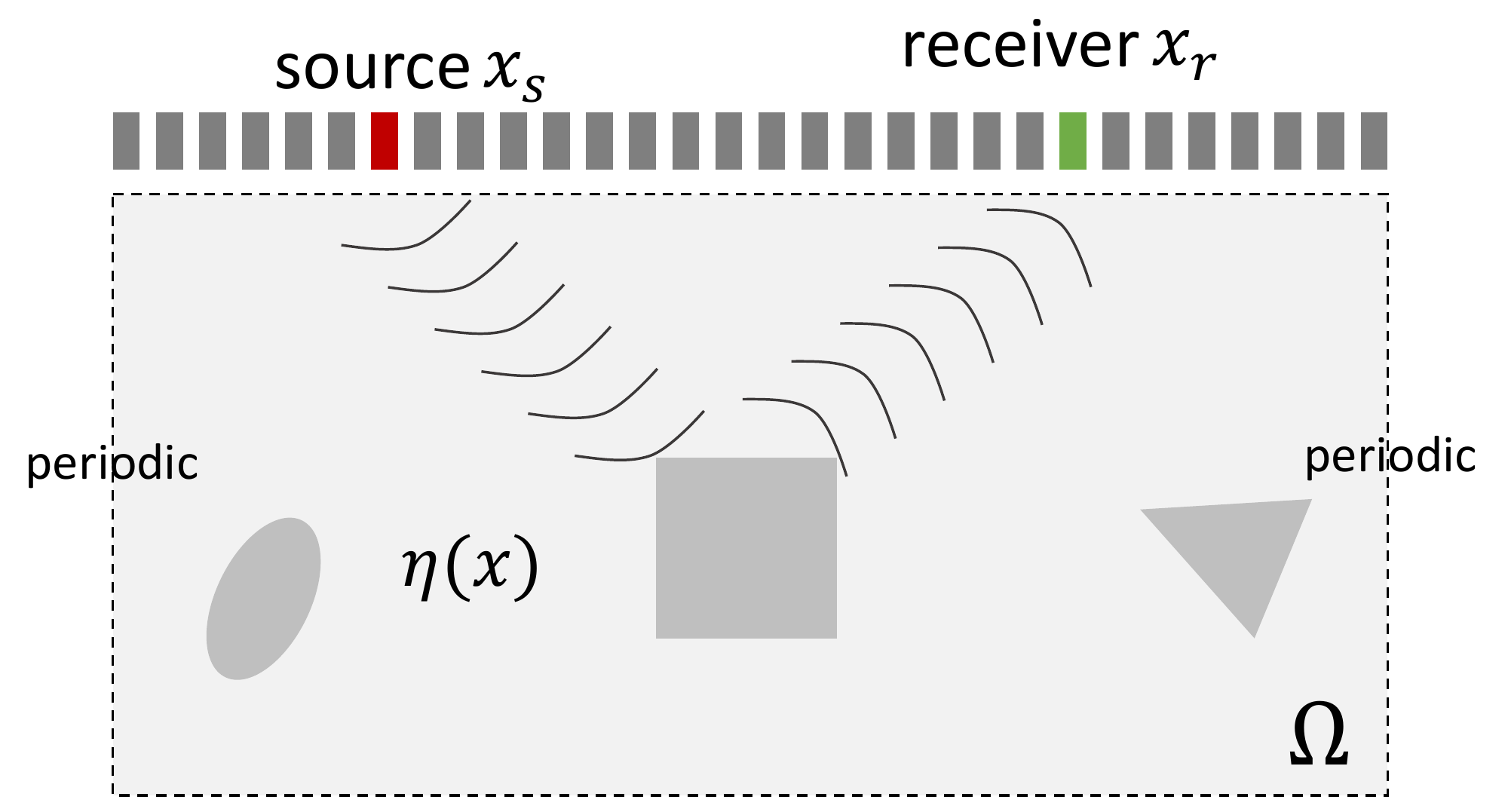}
  \caption{\label{fig:seismic} Illustration of a simple seismic imaging setting. The sources and
    receivers are located near the surface level (top) of the domain $\Omega$. The scatterer field
    $\eta(x)$ is assumed to be well-separated from the sources and the receivers.
  }
\end{figure}

Using the background velocity field $c_0(x)$, we first introduce the background Helmholtz operator
$L_0 = -\Delta - \omega^2/c_0(x)^2$.  For each source $s$, we place a delta source at point $x_s$
and solve the the Helmholtz equation (in the differential imaging setting)
\begin{equation} \label{eq:uss}
  (L_0 - \eta) (G_0(x, x_s) + u^s(x)) = 0,
\end{equation}
where $G_0 = L_0^{-1}$ be the Green's functions of the background Helmholtz operator $L_0$. The
solution is recorded at points $x_r$ for $r\in R$ and the whole dataset is $ d(s,r) \equiv
u^s(x_r)$. In order to understand better the relationship between $\eta(x)$ and $d(s,r)$, let us
perform a perturbative analysis for small $\eta$. Expanding \eqref{eq:uss} gives rise to
\[
(L_0 u^s)(x) = \eta(x) G_0(x, x_s) + \ldots.
\]
Solving this leads to
\[
d(s,r) = u^s(r) = \int G_0(x_r, x)G_0(x, x_s) \eta(x) \dd x + \ldots.
\]
Again, we introduce $d_1(s,r) = \int G_0(x_r,x)G_0(x,x_s) \eta(x) \dd x$ as the leading order linear
term in terms of $\eta$.

\begin{figure}[ht]
  \centering
  \begin{tabular}{c@{}c@{}c}
    \subfloat[$\eta(x)$]{ \includegraphics[width=0.3\textwidth]{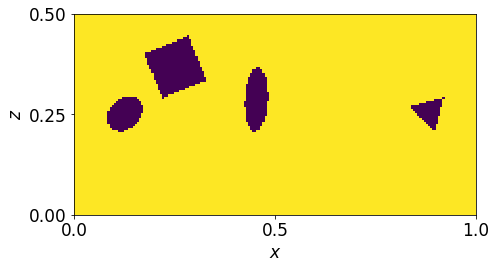} } &
    \subfloat[$\Re(d(s,r))$]{ \includegraphics[width=0.3\textwidth]{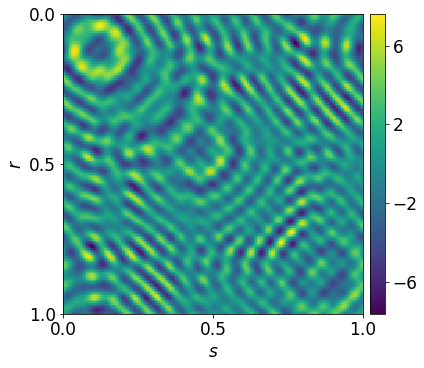} } &
    \subfloat[$\Im(d(s,r))$]{ \includegraphics[width=0.3\textwidth]{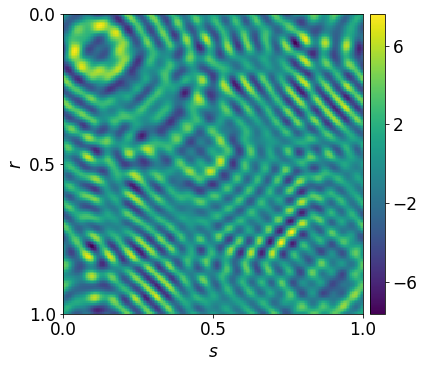} }\\
    &
    \subfloat[$\Re(d(m,h))$]{ \includegraphics[width=0.3\textwidth]{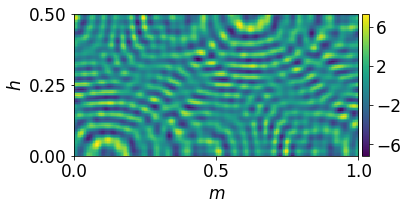} } &
    \subfloat[$\Im(d(m,h))$]{ \includegraphics[width=0.3\textwidth]{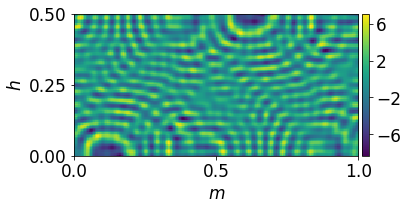} }
  \end{tabular}
  \caption{\label{fig:measurement_seismic} Visualization of the scatterer field $\eta$ and the
    measurement data $d$ for the seismic imaging.  The upper figures are the scatterer and the real
    and imaginary part of the measurement data $d$, respectively.  The lower two figures are the
    real and imaginary part of the measurement data after change of variable.  }
\end{figure}

\Cref{fig:measurement_seismic} gives an example of the scatterer field and the measurement data.
Notice that the strongest signal concentrates at the diagonal of the measurement data $d(s,r)$.
Because of the periodicity in the horizontal direction, it is convenient to rotate the measurement
data by a change of variables as
\begin{equation}
  m = \frac{r+s}{2}, h = \frac{r-s}{2},\quad\text{or equivalently}\quad r = m+h, s = m-h,
\end{equation}
where all the variables $m, h, r, s$ are understood modulus $1$.  With a bit abuse of notation, we
recast the measurement data
\begin{equation}
  d(m,h) \equiv d(s, r)|_{s=m-h, r=m+h},
\end{equation}
and so does for $d_1(m, h)$. At the same time, by letting $x=(p,z)$ where $p$ is horizontal
component of $x$ and $z$ is the depth component, we write $\eta(p,z) = \eta(x)$.  Since $d_1(m, h)$
is linearly dependent on $\eta(p, z)$, there exists a kernel distribution $K(m, h, p, z)$ such that
\begin{equation}
  d_1(m, h) = \int_{0}^Z\int_{0}^{1}K(m, h, p, z)\eta(p, z)\dd z\dd p.
\end{equation}

One of the most common scenario in seismic imaging is that $c_0(x)$ only depends on the depth, i.e.,
$c_0(p, z) \equiv c_0(z)$. Note that in this scenario the whole problem is equivariant to translation
in the horizontal direction. The system can be dramatically simplified due to the following
proposition.
\begin{proposition}\label{pro:convolution_seismic}
  There exists a function $\kappa(h, z, \cdot)$ periodic in the last parameter such that $K(m, h, p,
  z) = \kappa(h, z, m-p)$ or equivalently,
  \begin{equation}\label{eq:convolution_seismic}
    d_1(m,h) = \int_0^Z \int_0^{1} \kappa(h, z, m-p) \eta(p,z) \dd p \dd z.
  \end{equation}
\end{proposition}
\begin{proof}
  Because of $c_0(p, z) = c_0(z)$ and the periodic boundary conditions in the horizontal direction,
  the Green's function of the background Helmholtz operator $G_0$ is translation invariant on the
  horizontal direction, i.e., there exists a $g_0(\cdot,\cdot,\cdot)$ such that $G( (x_1, x_2),
  (y_1, y_2)) = g(x_1-y_1, x_2, y_2)$. Therefore, 
  \[
    \begin{aligned}
      d_1(m, h) &= \int_0^Z\int_0^1 G_0( (m+h, Z), (p, z)) G_0( (p, z), (m-h, Z)) \eta(p, z)\dd z \dd p \\
      &= \int_0^Z\int_0^1 g_0(m-p+h, Z, z)g_0(p-m+h, z, Z) \eta(p, z) \dd z\dd p.
    \end{aligned}
  \]
  Setting $\kappa(h, z, y) = g_0(y+h, Z, z)g_0(-y+h, z, Z)$ completes the proof.
\end{proof}

To discrete the problem, the scatterer $\eta(p, z)$ will be represented on a uniform mesh of
$[0,1]\times [0,Z]$.  With a slight abuse of notation, we shall use the same symbols to denote the
discretization version of the continuous kernels and variables. The discrete version of
\cref{eq:convolution_seismic} then becomes
\begin{equation}\label{eq:discrete_seismic}
  d(m,h) \approx \sum_{z}\left( \kappa(h,z,\cdot) * \eta(\cdot, z) \right)(m).
\end{equation}

\subsection{Neural network and numerical examples}

\subsubsection{Neural network}

Note that the key of the neural network architecture in \cref{alg:inverse} for the far field pattern
case is the convolution form in the angular direction in \cref{pro:convolution}. For the seismic
imaging case, \cref{pro:convolution_seismic} is the counterpart of \cref{pro:convolution}. Since the
argument in \cref{sec:NN} remains valid for seismic imaging, the neural network architecture for
seismic imaging is the same as that in \cref{alg:inverse}. However, the hyper-parameters are
problem-dependent.

\subsubsection{Experimental setup}
In the experiment $Z=1/2$ and the domain $\Omega=[0,1]\times [0, Z]$ is discretized with a uniform
Cartesian mesh with $192\times 96$ points with frequency $\omega=16$. The remaining setup of the
numerical solution of the Helmholtz equation is same as that for the far field pattern problem. For
the measurement, we also set the number of sources and receivers as $N_s=N_r=192$. The measurement
data $d(s,r)$ is generated by solving the Helmholtz equation $N_s$ times by placing a delta function
on each source point.  For the change of variable of $(s, r)\to (m, h)$, linear interpolation is
used for generating the data $d(m,h)$ from $d(s, r)$, with $N_m=192$ for $m\in[0,1)$ and $N_h=96$
  for $h\in(0, 1/2)$. In the actual simulation, we use both the real and imaginary part and
  concentrate them on the $h$ direction as the input.

\subsubsection{Results}

The numerical experiments here focus on the shape reconstruction setting, where $\eta$ are piecewise
constant inclusions.  Here, the scatterer field $\eta$ is assumed to be the sum of $\Nshape$
piecewise constant shapes.  For each shape, it can be either triangle, square or ellipse, the
orientation is uniformly random over the unit circle, the position is uniformly sampled in the
$[0,1]\times[0.2,0.4]$, and the circumradius is sampled from the uniform distribution
$\mU(0.1,0.2)$. If the shape is ellipse, its width and height are sampled from the uniform
distribution $\mU(0.08, 0.16)$ and $\mU(0.04, 0.8)$.  It is also required that there is no
intersection between any two shapes. We generate two datasets with $\Nshape=2$ and $\Nshape=4$ and
each has $20,480$ samples $\{(\eta_i, d_i)\}$ with $16,384$ used for training and the remaining
$4,096$ reserved for testing.

\begin{figure}[h!]
  \centering
  \includegraphics[width=0.4\textwidth]{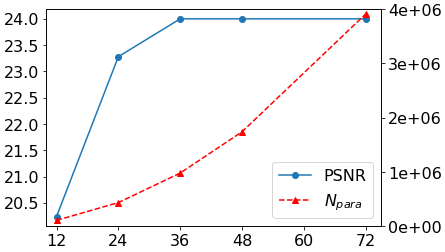}
  \caption{\label{fig:trErr_seismic} The test PSNR for seismic imaging for different channel
    numbers $c$ with $\Nshape=4$.  
  }
\end{figure}

\begin{figure}[h!]
  \centering
  \begin{tabular}{:l@{}cc:}
	\hdashline
    \multirow{2}{*}{\rotatebox{90}{\phantom{abcde}$\Nshape=4$}} &
    \subfloat[reference]{ \includegraphics[width=0.44\textwidth]{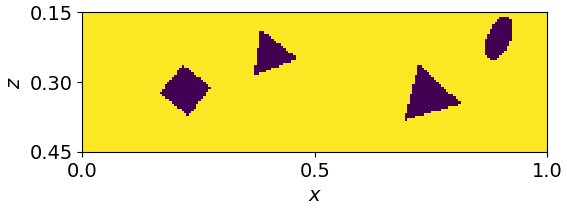} }       &
    \subfloat[$\eta^{\NN}$ with $\delta=0$]{ \includegraphics[width=0.44\textwidth]{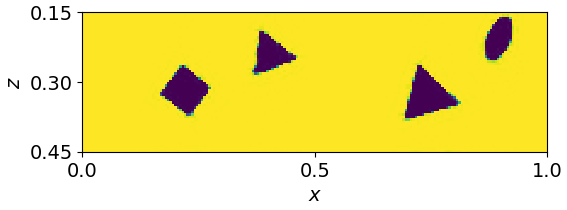} }       \\
    &
    \subfloat[$\eta^{\NN}$ with $\delta=10\%$]{ \includegraphics[width=0.44\textwidth]{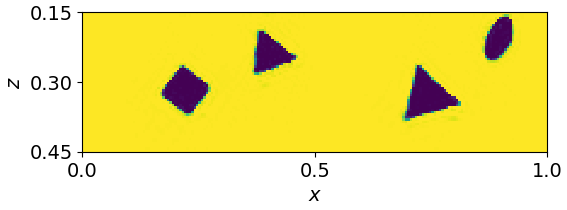} }       &
    \subfloat[$\eta^{\NN}$ with $\delta=100\%$]{ \includegraphics[width=0.44\textwidth]{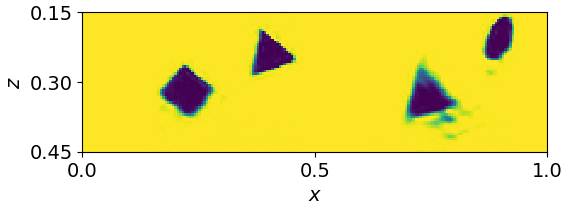} }    \\
	\hdashline
    \multirow{2}{*}{\rotatebox{90}{\phantom{abcde}$\Nshape=2$}} &
    \subfloat[reference]{ \includegraphics[width=0.44\textwidth]{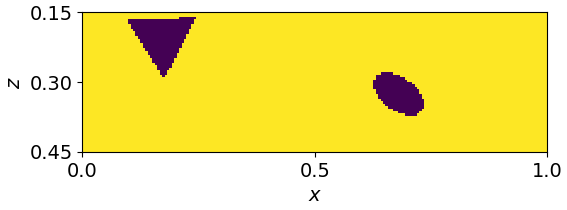} }       &
    \subfloat[$\eta^{\NN}$ with $\delta=0$]{ \includegraphics[width=0.44\textwidth]{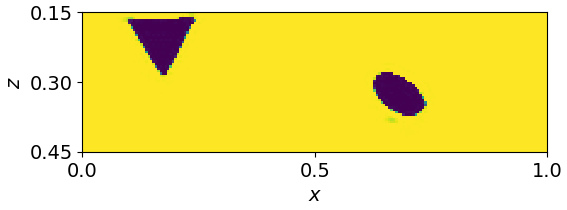} }       \\
    &
    \subfloat[$\eta^{\NN}$ with $\delta=10\%$]{ \includegraphics[width=0.44\textwidth]{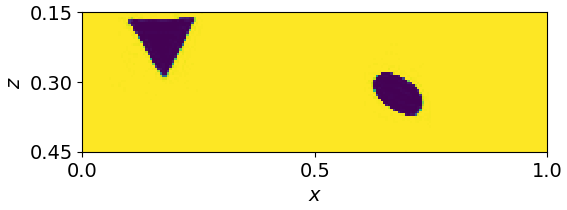} }       &
    \subfloat[$\eta^{\NN}$ with $\delta=100\%$]{ \includegraphics[width=0.44\textwidth]{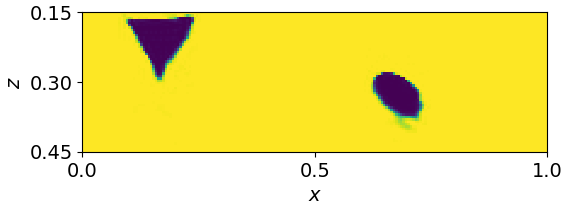} }   \\
	\hdashline
  \end{tabular}
  \caption{\label{fig:prediction_seismic} The seismic imaging problem. NN prediction of a sample
    with test data $\Nshape=4$ (the first two rows) or $\Nshape=2$ (the last two rows), at different
    noise level $\delta=0$, $10\%$ and $100\%$.  }
\end{figure}

\begin{figure}[h!]
  \centering
  \begin{tabular}{:l@{}cc:}
	\hdashline
    	\rotatebox{90}{Test:\phantom{n} $\Nshape=4$} &
    \subfloat[reference]{ \includegraphics[width=0.44\textwidth]{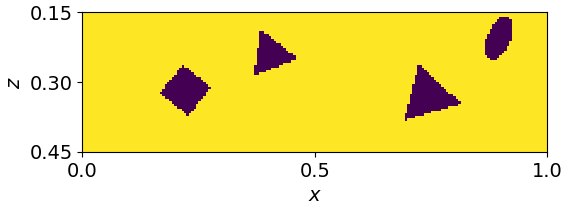} }       &
    \subfloat[$\eta^{\NN}$ with $\delta=0$]{ \includegraphics[width=0.44\textwidth]{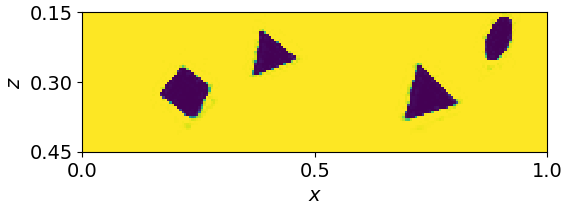} }       \\
	\rotatebox{90}{Train: $\Nshape=2$;} &
    \subfloat[$\eta^{\NN}$ with $\delta=10\%$]{ \includegraphics[width=0.44\textwidth]{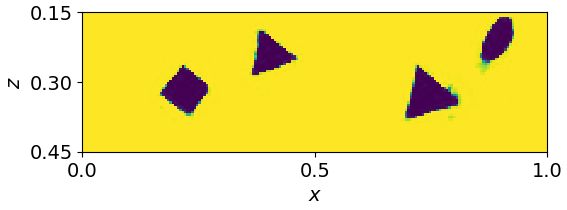} }       &
    \subfloat[$\eta^{\NN}$ with $\delta=100\%$]{ \includegraphics[width=0.44\textwidth]{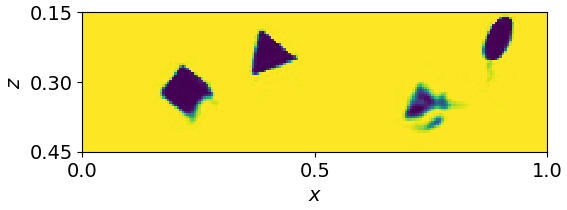} }    \\
	\hdashline
    	\rotatebox{90}{Test:\phantom{n} $\Nshape=2$} &
    \subfloat[reference]{ \includegraphics[width=0.44\textwidth]{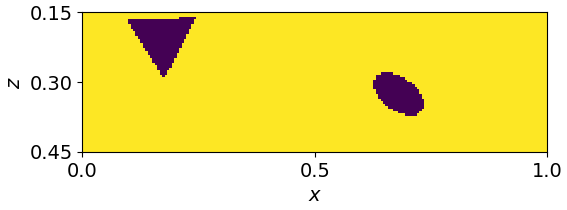} }       &
    \subfloat[$\eta^{\NN}$ with $\delta=0$]{ \includegraphics[width=0.44\textwidth]{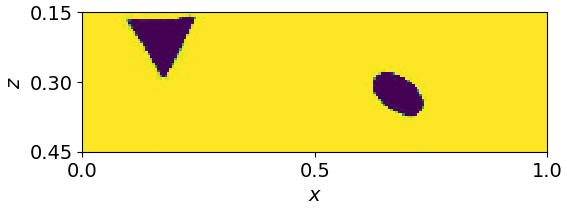} }       \\
	\rotatebox{90}{Train: $\Nshape=4$;} &
    \subfloat[$\eta^{\NN}$ with $\delta=10\%$]{ \includegraphics[width=0.44\textwidth]{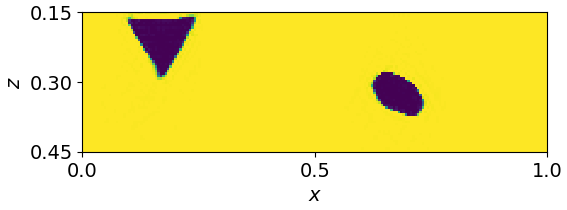} }       &
    \subfloat[$\eta^{\NN}$ with $\delta=100\%$]{ \includegraphics[width=0.44\textwidth]{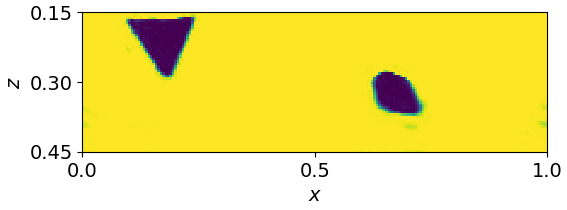} }    \\
	\hdashline
  \end{tabular}
  \caption{\label{fig:generalization_seismic} The seismic imaging problem: NN generalization test.
    In the top figures, the NN is trained by the data with $\Nshape=2$ at noise level $\delta=0$,
    $10\%$ or $100\%$ and tested by the data with $\Nshape=4$. In the bottom figures, the NN is
    trained by the data with $\Nshape=4$ and tested by the data with $\Nshape=2$.  }
\end{figure}

The first study is about the choice of channel number $c$ in
\cref{alg:inverse}. \Cref{fig:trErr} presents the test PSNR and the number of parameters, for
different channel number $c$ on the dataset $\Nshape=4$.  Similar to the far field pattern problem,
as the channel number $c$ increases, the test PSNR first consistently increases and then saturates.
Notice that the number of parameters of the neural network is $O(c^2\log(N_{\theta})N_{\cnn})$. The
choice of $c=36$ is a reasonable balance between accuracy and efficiency and the total number of
parameters is $981$K.

To model the uncertainty in the measurement data, the same method as the far field pattern problem
is used to add noises to the measurement data.  \Cref{fig:prediction_seismic} collects, for
different noise level $\delta=0$, $10\%$, $100\%$, samples for $\Nshape=2$ and $4$, and
\cref{fig:generalization_seismic} presents the generalization test of the proposed NN by training
and testing on different datasets..

\section{Discussions}\label{sec:conclusion}
This paper presents a neural network approach for the two typical problems of the inverse
scattering: far field pattern and seismic imaging. The approach uses the NN to approximate the whole
inverse map from the measurement data to the scatterer field, inspired by the perturbative analysis
that indicates that the linearized forward map can be represented by a one-dimensional convolution
with multiple channels. The analysis in this paper can also be extended to three-dimensional
scattering problems. The analysis of seismic imaging can be easily extended to non-periodic boundary
conditions by replacing the periodic padding in \cref{alg:inverse} with zero padding.

\section*{Acknowledgments}
The work of Y.F. and L.Y. is partially supported by the U.S. Department of Energy, Office of
Science, Office of Advanced Scientific Computing Research, Scientific Discovery through Advanced
Computing (SciDAC) program. The work of L.Y. is also partially supported by the National Science
Foundation under award DMS-1818449.

\bibliographystyle{abbrv}
\bibliography{nn}

\def\cprime{$'$}
\begin{thebibliography}{10}

\bibitem{tensorflow}
M.~Abadi et~al.
\newblock Tensorflow: A system for large-scale machine learning.
\newblock In {\em OSDI}, volume~16, pages 265--283, 2016.

\bibitem{adler2017solving}
J.~Adler and O.~{\"O}ktem.
\newblock Solving ill-posed inverse problems using iterative deep neural
  networks.
\newblock {\em Inverse Problems}, 33(12):124007, 2017.

\bibitem{Araya-Polo2018}
M.~Araya-Polo, J.~Jennings, A.~Adler, and T.~Dahlke.
\newblock Deep-learning tomography.
\newblock {\em The Leading Edge}, 37(1):58--66, 2018.

\bibitem{bar2019unsupervised}
L.~Bar and N.~Sochen.
\newblock Unsupervised deep learning algorithm for {PDE}-based forward and
  inverse problems.
\newblock {\em arXiv preprint arXiv:1904.05417}, 2019.

\bibitem{berenger1996perfectly}
J.-P. Berenger.
\newblock Perfectly matched layer for the {FDTD} solution of wave-structure
  interaction problems.
\newblock {\em IEEE Transactions on antennas and propagation}, 44(1):110--117,
  1996.

\bibitem{berg2017unified}
J.~Berg and K.~Nystr{\"o}m.
\newblock A unified deep artificial neural network approach to partial
  differential equations in complex geometries.
\newblock {\em Neurocomputing}, 317:28--41, 2018.

\bibitem{bcr}
G.~Beylkin, R.~Coifman, and V.~Rokhlin.
\newblock Fast wavelet transforms and numerical algorithms {I}.
\newblock {\em Communications on pure and applied mathematics}, 44(2):141--183,
  1991.

\bibitem{borden2001mathematical}
B.~Borden.
\newblock Mathematical problems in radar inverse scattering.
\newblock {\em Inverse Problems}, 18(1):R1, 2001.

\bibitem{cakoni2011linear}
F.~Cakoni, D.~Colton, and P.~Monk.
\newblock {\em The linear sampling method in inverse electromagnetic
  scattering}, volume~80.
\newblock SIAM, 2011.

\bibitem{carleo2017solving}
G.~Carleo and M.~Troyer.
\newblock Solving the quantum many-body problem with artificial neural
  networks.
\newblock {\em Science}, 355(6325):602--606, 2017.

\bibitem{cheney2001linear}
M.~Cheney.
\newblock The linear sampling method and the music algorithm.
\newblock {\em Inverse problems}, 17(4):591, 2001.

\bibitem{keras}
F.~Chollet et~al.
\newblock Keras.
\newblock \url{https://keras.io}, 2015.

\bibitem{colton2018looking}
D.~Colton and R.~Kress.
\newblock Looking back on inverse scattering theory.
\newblock {\em SIAM Review}, 60(4):779--807, 2018.

\bibitem{colton1998inverse}
D.~L. Colton, R.~Kress, and R.~Kress.
\newblock {\em Inverse acoustic and electromagnetic scattering theory},
  volume~93.
\newblock Springer, 1998.

\bibitem{cybenko1989approximation}
G.~Cybenko.
\newblock Approximation by superpositions of a sigmoidal function.
\newblock {\em Mathematics of control, signals and systems}, 2(4):303--314,
  1989.

\bibitem{dozat2015incorporating}
T.~Dozat.
\newblock Incorporating {N}esterov momentum into adam.
\newblock {\em International Conference on Learning Representations}, 2016.

\bibitem{weinan2018deep}
W.~E and B.~Yu.
\newblock The deep {R}itz method: A deep learning-based numerical algorithm for
  solving variational problems.
\newblock {\em Communications in Mathematics and Statistics}, 6(1):1--12, 2018.

\bibitem{engquist2010fast}
B.~Engquist and L.~Ying.
\newblock Fast directional algorithms for the {H}elmholtz kernel.
\newblock {\em Journal of Computational and Applied Mathematics},
  234(6):1851--1859, 2010.

\bibitem{engquist2011hmatrix}
B.~Engquist and L.~Ying.
\newblock Sweeping preconditioner for the {H}elmholtz equation: hierarchical
  matrix representation.
\newblock {\em Communications on pure and applied mathematics}, 64(5):697--735,
  2011.

\bibitem{engquist2011pml}
B.~Engquist and L.~Ying.
\newblock Sweeping preconditioner for the {H}elmholtz equation: moving
  perfectly matched layers.
\newblock {\em Multiscale Modeling \& Simulation}, 9(2):686--710, 2011.

\bibitem{erlangga2008advances}
Y.~A. Erlangga.
\newblock Advances in iterative methods and preconditioners for the {H}elmholtz
  equation.
\newblock {\em Archives of Computational Methods in Engineering}, 15(1):37--66,
  2008.

\bibitem{ernst2012difficult}
O.~G. Ernst and M.~J. Gander.
\newblock Why it is difficult to solve {H}elmholtz problems with classical
  iterative methods.
\newblock In {\em Numerical analysis of multiscale problems}, pages 325--363.
  Springer, 2012.

\bibitem{fan2019bcr}
Y.~Fan, C.~O. Bohorquez, and L.~Ying.
\newblock {BCR-Net}: a neural network based on the nonstandard wavelet form.
\newblock {\em Journal of Computational Physics}, 384:1--15, 2019.

\bibitem{fan2018mnnh2}
Y.~Fan, J.~Feliu-Fab{\`a}, L.~Lin, L.~Ying, and L.~Zepeda-N{\'u}\~nez.
\newblock A multiscale neural network based on hierarchical nested bases.
\newblock {\em Research in the Mathematical Sciences}, 6(2):21, 2019.

\bibitem{fan2018mnn}
Y.~Fan, L.~Lin, L.~Ying, and L.~Zepeda-N\'u\~nez.
\newblock A multiscale neural network based on hierarchical matrices.
\newblock {\em arXiv preprint arXiv:1807.01883}, 2018.

\bibitem{fan2019eit}
Y.~Fan and L.~Ying.
\newblock Solving electrical impedance tomography with deep learning.
\newblock {\em arXiv preprint arXiv:1906.03944}, 2019.

\bibitem{fan2019ot}
Y.~Fan and L.~Ying.
\newblock Solving optical tomography with deep learning.
\newblock {\em arXiv preprint arXiv:1910.04756}, 2019.

\bibitem{feliu2019meta}
J.~Feliu-Faba, Y.~Fan, and L.~Ying.
\newblock Meta-learning pseudo-differential operators with deep neural
  networks.
\newblock {\em arXiv preprint arXiv:1906.06782}, 2019.

\bibitem{gander2001ailu}
M.~J. Gander and F.~Nataf.
\newblock {AILU} for {H}elmholtz problems: a new preconditioner based on the
  analytic parabolic factorization.
\newblock {\em Journal of Computational Acoustics}, 9(04):1499--1506, 2001.

\bibitem{glorot2010understanding}
X.~Glorot and Y.~Bengio.
\newblock Understanding the difficulty of training deep feedforward neural
  networks.
\newblock In {\em Proceedings of the thirteenth international conference on
  artificial intelligence and statistics}, pages 249--256, 2010.

\bibitem{goodfellow2016deep}
I.~Goodfellow, Y.~Bengio, A.~Courville, and Y.~Bengio.
\newblock {\em Deep learning}, volume~1.
\newblock MIT press Cambridge, 2016.

\bibitem{greene1988acoustical}
C.~Greene, P.~Wiebe, J.~Burczynski, and M.~Youngbluth.
\newblock Acoustical detection of high-density krill demersal layers in the
  submarine canyons off georges bank.
\newblock {\em Science}, 241(4863):359--361, 1988.

\bibitem{han2018solving}
J.~Han, A.~Jentzen, and W.~E.
\newblock Solving high-dimensional partial differential equations using deep
  learning.
\newblock {\em Proceedings of the National Academy of Sciences},
  115(34):8505--8510, 2018.

\bibitem{han2017deep}
J.~Han, L.~Zhang, R.~Car, and W.~E.
\newblock Deep potential: A general representation of a many-body potential
  energy surface.
\newblock {\em Communications in Computational Physics}, 23(3):629--639, 2018.

\bibitem{henriksson2010quantitative}
T.~Henriksson, N.~Joachimowicz, C.~Conessa, and J.-C. Bolomey.
\newblock Quantitative microwave imaging for breast cancer detection using a
  planar 2.45 ghz system.
\newblock {\em IEEE Transactions on Instrumentation and Measurement},
  59(10):2691--2699, 2010.

\bibitem{Hinton2012}
G.~Hinton, L.~Deng, D.~Yu, G.~E. Dahl, A.~r.~Mohamed, N.~Jaitly, A.~Senior,
  V.~Vanhoucke, P.~Nguyen, T.~N. Sainath, and B.~Kingsbury.
\newblock Deep neural networks for acoustic modeling in speech recognition: The
  shared views of four research groups.
\newblock {\em IEEE Signal Processing Magazine}, 29(6):82--97, 2012.

\bibitem{hoole1993artificial}
S.~R.~H. Hoole.
\newblock Artificial neural networks in the solution of inverse electromagnetic
  field problems.
\newblock {\em IEEE transactions on Magnetics}, 29(2):1931--1934, 1993.

\bibitem{hulme2004general}
T.~Hulme, A.~Haines, and J.~Yu.
\newblock General elastic wave scattering problems using an impedance operator
  approach-ii. two-dimensional isotropic validation and examples.
\newblock {\em Geophysical Journal International}, 159(2):658--666, 2004.

\bibitem{hussein2014dynamics}
M.~I. Hussein, M.~J. Leamy, and M.~Ruzzene.
\newblock Dynamics of phononic materials and structures: Historical origins,
  recent progress, and future outlook.
\newblock {\em Applied Mechanics Reviews}, 66(4):040802, 2014.

\bibitem{kabir2008neural}
H.~Kabir, Y.~Wang, M.~Yu, and Q.-J. Zhang.
\newblock Neural network inverse modeling and applications to microwave filter
  design.
\newblock {\em IEEE Transactions on Microwave Theory and Techniques},
  56(4):867--879, 2008.

\bibitem{khoo2017solving}
Y.~Khoo, J.~Lu, and L.~Ying.
\newblock Solving parametric {PDE} problems with artificial neural networks.
\newblock {\em arXiv preprint arXiv:1707.03351}, 2017.

\bibitem{khoo2019committor}
Y.~Khoo, J.~Lu, and L.~Ying.
\newblock Solving for high-dimensional committor functions using artificial
  neural networks.
\newblock {\em Research in the Mathematical Sciences}, 6(1):1, 2019.

\bibitem{khoo2018switchnet}
Y.~Khoo and L.~Ying.
\newblock {SwitchNet}: a neural network model for forward and inverse
  scattering problems.
\newblock {\em arXiv preprint arXiv:1810.09675}, 2018.

\bibitem{kirsch1999factorization}
A.~Kirsch.
\newblock Factorization of the far-field operator for the inhomogeneous medium
  case and an application in inverse scattering theory.
\newblock {\em Inverse problems}, 15(2):413, 1999.

\bibitem{kirsch2008factorization}
A.~Kirsch and N.~Grinberg.
\newblock {\em The factorization method for inverse problems}, volume~36.
\newblock Oxford University Press, 2008.

\bibitem{Krizhevsky2012}
A.~Krizhevsky, I.~Sutskever, and G.~E. Hinton.
\newblock Image{N}et classification with deep convolutional neural networks.
\newblock In {\em Proceedings of the 25th International Conference on Neural
  Information Processing Systems - Volume 1}, NIPS'12, pages 1097--1105, USA,
  2012. Curran Associates Inc.

\bibitem{kutyniok2019theoretical}
G.~Kutyniok, P.~Petersen, M.~Raslan, and R.~Schneider.
\newblock A theoretical analysis of deep neural networks and parametric {PDEs}.
\newblock {\em arXiv preprint arXiv:1904.00377}, 2019.

\bibitem{leCunn2015}
Y.~LeCun, Y.~Bengio, and G.~Hinton.
\newblock Deep learning.
\newblock {\em Nature}, 521(436), 2015.

\bibitem{Leung2014}
M.~K.~K. Leung, H.~Y. Xiong, L.~J. Lee, and B.~J. Frey.
\newblock Deep learning of the tissue-regulated splicing code.
\newblock {\em Bioinformatics}, 30(12):i121--i129, 2014.

\bibitem{li2019variational}
Y.~Li, J.~Lu, and A.~Mao.
\newblock Variational training of neural network approximations of solution
  maps for physical models.
\newblock {\em arXiv preprint arXiv:1905.02789}, 2019.

\bibitem{fei2016recursive}
F.~Liu and L.~Ying.
\newblock Recursive sweeping preconditioner for the three-dimensional
  {H}elmholtz equation.
\newblock {\em SIAM Journal on Scientific Computing}, 38(2):A814--A832, 2016.

\bibitem{long2018pde}
Z.~Long, Y.~Lu, X.~Ma, and B.~Dong.
\newblock {PDE}-net: Learning {PDE}s from data.
\newblock In J.~Dy and A.~Krause, editors, {\em Proceedings of the 35th
  International Conference on Machine Learning}, volume~80 of {\em Proceedings
  of Machine Learning Research}, pages 3208--3216, Stockholmsmässan, Stockholm
  Sweden, 10--15 Jul 2018. PMLR.

\bibitem{lucas2018using}
A.~Lucas, M.~Iliadis, R.~Molina, and A.~K. Katsaggelos.
\newblock Using deep neural networks for inverse problems in imaging: beyond
  analytical methods.
\newblock {\em IEEE Signal Processing Magazine}, 35(1):20--36, 2018.

\bibitem{MaSheridan2015}
J.~Ma, R.~P. Sheridan, A.~Liaw, G.~E. Dahl, and V.~Svetnik.
\newblock Deep neural nets as a method for quantitative structure-activity
  relationships.
\newblock {\em Journal of Chemical Information and Modeling}, 55(2):263--274,
  2015.

\bibitem{mora1987nonlinear}
P.~Mora.
\newblock Nonlinear two-dimensional elastic inversion of multioffset seismic
  data.
\newblock {\em Geophysics}, 52(9):1211--1228, 1987.

\bibitem{norton1979ultrasonic}
S.~J. Norton and M.~Linzer.
\newblock Ultrasonic reflectivity tomography: reconstruction with circular
  transducer arrays.
\newblock {\em Ultrasonic imaging}, 1(2):154--184, 1979.

\bibitem{palermo2016engineered}
A.~Palermo, S.~Kr{\"o}del, A.~Marzani, and C.~Daraio.
\newblock Engineered metabarrier as shield from seismic surface waves.
\newblock {\em Scientific reports}, 6:39356, 2016.

\bibitem{Raissi2018}
M.~Raissi and G.~E. Karniadakis.
\newblock Hidden physics models: Machine learning of nonlinear partial
  differential equations.
\newblock {\em Journal of Computational Physics}, 357:125 -- 141, 2018.

\bibitem{raissi2019physics}
M.~Raissi, P.~Perdikaris, and G.~E. Karniadakis.
\newblock Physics-informed neural networks: A deep learning framework for
  solving forward and inverse problems involving nonlinear partial differential
  equations.
\newblock {\em Journal of Computational Physics}, 378:686--707, 2019.

\bibitem{rudd2015constrained}
K.~Rudd and S.~Ferrari.
\newblock A constrained integration ({CINT}) approach to solving partial
  differential equations using artificial neural networks.
\newblock {\em Neurocomputing}, 155:277--285, 2015.

\bibitem{SCHMIDHUBER2015}
J.~Schmidhuber.
\newblock Deep learning in neural networks: An overview.
\newblock {\em Neural Networks}, 61:85--117, 2015.

\bibitem{SutskeverNIPS2014}
I.~Sutskever, O.~Vinyals, and Q.~V. Le.
\newblock Sequence to sequence learning with neural networks.
\newblock In Z.~Ghahramani, M.~Welling, C.~Cortes, N.~D. Lawrence, and K.~Q.
  Weinberger, editors, {\em Advances in Neural Information Processing Systems
  27}, pages 3104--3112. Curran Associates, Inc., 2014.

\bibitem{tan2018image}
C.~Tan, S.~Lv, F.~Dong, and M.~Takei.
\newblock Image reconstruction based on convolutional neural network for
  electrical resistance tomography.
\newblock {\em IEEE Sensors Journal}, 19(1):196--204, 2018.

\bibitem{verschuur1997estimation}
D.~Verschuur and A.~Berkhout.
\newblock Estimation of multiple scattering by iterative inversion, part ii:
  Practical aspects and examples.
\newblock {\em Geophysics}, 62(5):1596--1611, 1997.

\bibitem{weglein2003inverse}
A.~B. Weglein, F.~V. Ara{\'u}jo, P.~M. Carvalho, R.~H. Stolt, K.~H. Matson,
  R.~T. Coates, D.~Corrigan, D.~J. Foster, S.~A. Shaw, and H.~Zhang.
\newblock Inverse scattering series and seismic exploration.
\newblock {\em Inverse problems}, 19(6):R27, 2003.

\end{thebibliography}
\end{document}